\documentclass[a4paper,12pt]{elsarticle}
\usepackage[T1]{fontenc}
\usepackage[english]{babel}
\usepackage{ae,aecompl}
\usepackage{amsmath, amsthm}
\usepackage{amssymb}
\usepackage[utf8]{inputenc}
\usepackage[section] {placeins}
\usepackage[ruled, vlined, linesnumbered]{algorithm2e}
\newtheorem {Theorem}                 {Theorem}         [section]

\newtheorem {myalgorithm}    [Theorem]  {Algorithm}

\newtheorem {lemma}        [Theorem]  {Lemma}

\input amssym.def
\input amssym

\usepackage{pgfplots}
\usepackage{verbatim}
\usepackage{subfigure}
\usepackage{tikz}
\usetikzlibrary{shapes,arrows,backgrounds,mindmap}
\usepackage{titlesec}
\journal{arXiv}

\begin{document}
	\begin{frontmatter}
		\title{$2$-edge-twinless blocks}
		\author{Raed Jaberi}
		
		\begin{abstract}  
		Let $G=(V,E)$ be a directed graph. A $2$-edge-twinless block in $G$ is a maximal vertex set $C^{t}\subseteq V$ with $|C^{t}|>1$ such that for any distinct vertices $v,w \in C^{t}$, and for every edge $e\in E$, the vertices $v,w$ are in the same twinless strongly connected component of $G\setminus\left \lbrace e \right\rbrace $. 
		In this paper we study this concept and describe algorithms for computing  $2$-edge-twinless blocks.
		\end{abstract} 
	\end{frontmatter}
	\section{Introduction}
Let $G=(V,E))$ be a directed graph.
	A \textit{$2$-edge block} in $G$ is a maximal vertex set $C^{e}\subseteq V$ with $|C^{e}|\geq 2$ such that for each pair of distinct vertices $v, w \in C^{e}$, there are two edge-disjoint paths from $v$ to $w$ and two edge-disjoint paths from $w$ to $v$ in $G$. By Menger's Theorem for edge connectivity \cite{Menger27},  there exist two edge-disjoint paths from $v$ to $w$ and two edge-disjoint paths from $w$ to $v$ in $G$ if and only if the vertices $v,w$ belong to the same strongly connected component of $G\setminus\left\lbrace e\right\rbrace $ for any edge $e\in E$. In this paper we introduce and study a new concept the $2$-edge-twinless blocks. A $2$-edge-twinless block in $G$ is a maximal vertex set $C^{t}\subseteq V$ of size at least $2$ such that for any distinct vertices $v,w \in C$, and for every edge $e\in E$, the vertices $v,w$ lie in the same twinless strongly connected component of $G\setminus\left \lbrace e \right\rbrace $. As Figure \ref{fig:2edgetwinlessblocksexample} demonstrates, the vertices of a $2$-edge-twinless block do not necessarily lie in the same $2$-edge block.

\begin{figure}[h]
	\centering
	\scalebox{0.96}{
		\begin{tikzpicture}[xscale=2]
		\tikzstyle{every node}=[color=black,draw,circle,minimum size=0.9cm]
		\node (v1) at (-1.2,3.1) {$1$};
		\node (v2) at (-2.5,0) {$2$};
		\node (v3) at (-0.5, -2.5) {$3$};
		\node (v4) at (0,-0.5) {$4$};
		\node (v5) at (0.5,3.6) {$5$};
		\node (v6) at (0.6,1) {$6$};
		\node (v7) at (2.9,2.5) {$7$};
		\node (v8) at (1.2,-2.9) {$8$};
		\node (v9) at (-0.4,1.7) {$9$};
		\node (v10) at (2.5,-2.1) {$10$};
	
		\node (v11) at (1.9,7.6) {$11$};
		\node (v12) at (-2.5,4.9) {$12$};
		\node (v13) at (-1, 4.2) {$13$};
		\node (v14) at (0,4.5) {$14$};
		\node (v15) at (-1.2,-1) {$15$};
		\node (v16) at (0,5.7) {$16$};
		\node (v17) at (1.2,6.9) {$17$};
		\node (v18) at (2.9,5.5) {$18$};
		\node (v19) at (-0.4,7.7) {$19$};

		\begin{scope}   
		\tikzstyle{every node}=[auto=right]   
		\draw [-triangle 45] (v2) to (v5);
		\draw [-triangle 45] (v15) to (v3);
		\draw [-triangle 45] (v2) to (v1);
		\draw [-triangle 45] (v7) to (v3);
		\draw [-triangle 45] (v5) to (v7);
		\draw [-triangle 45] (v5) to (v9);
		\draw [-triangle 45] (v7) to[bend right] (v5);
	\draw [-triangle 45] (v9) to (v2);
	\draw [-triangle 45] (v8) to (v4);
	\draw [-triangle 45] (v4) to (v6);
	\draw [-triangle 45] (v2) to (v15);
	\draw [-triangle 45] (v10) to (v7);
	\draw [-triangle 45] (v3) to (v8);
	\draw [-triangle 45] (v6) to (v2);
	\draw [-triangle 45] (v8) to (v10);
	\draw [-triangle 45] (v1) to (v12);
	\draw [-triangle 45] (v12) to(v2);
	\draw [-triangle 45] (v18) to (v17);

	\draw [-triangle 45] (v12) to [bend left ](v19);
	\draw [-triangle 45] (v12) to (v16);
	\draw [-triangle 45] (v13) to (v12);

	\draw [-triangle 45] (v16) to (v18);
	\draw [-triangle 45] (v18) to (v14);
	\draw [-triangle 45] (v14) to (v13);
	\draw [-triangle 45] (v17) to[bend right ] (v12);
	\draw [-triangle 45] (v19) to (v11);
	\draw [-triangle 45] (v11) to (v18);
		\end{scope}
		\end{tikzpicture}}
	\caption{A twinless strongly connected graph $G$. This graph contains two $2$-edges blocks $C_{1}=\left\lbrace 12,18\right\rbrace , C_{2}=\left\lbrace 2,7\right\rbrace $ and one $2$-edge-twinless block $C=\left\lbrace 12,18 \right\rbrace $. Since the vertices $2,7$ are not in the same twinless strongly connected component of $G\setminus\left\lbrace (3,8) \right\rbrace $, the set $C_{2}$ is not a $2$-edge-twinless block.}
	\label{fig:2edgetwinlessblocksexample}
\end{figure}
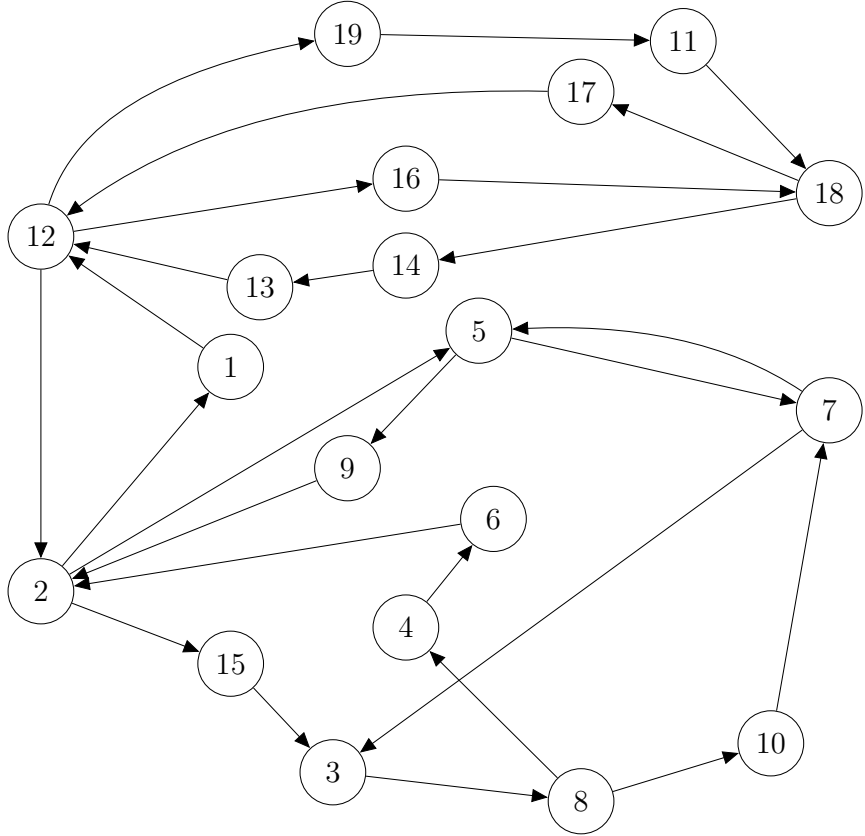

The concept of twinless strongly connected components was first introduced by	Raghavan \cite{SR06} in $2006$. Raghavan \cite{SR06} proved that twinless strongly connected components of a directed graph can be identified in linear time. The first linear time algorithm for testing $2$-vertex connectivity of directed graphs was described by Georgiadis \cite{G10} in $2010$. Italiano et al. \cite{ILS12,Italiano2010} gave linear time algorithms for determining all the strong articulation points and strong bridges of a directed graph. In $2014$, Jaberi \cite{J14} presented algorithms for computing the $2$-vertex-connected components of directed graphs in $O(nm)$ time (published in \cite{Jaberi16}). An experimental study ($2015$) \cite{LGILP15} showed that our  algorithm performs well in practice. Henzinger et al.\cite{HKL15}
introduced algorithms for calculating the $2$-vertex-connected components and $2$-edge-connected components of a directed graph in $O(n^{2})$ time. Jaberi \cite{Jaberi15} presented algorithms for computing the $2$-directed blocks, $2$-strong blocks, and $2$-edge blocks of a directed graph. Georgiadis et al. \cite{GILP14SODA} gave linear time algorithms for finding $2$-edge blocks. The same authors \cite{GILP14VertexConnectivity} gave linear time algorithms for calculating $2$-directed blocks and $2$-strong blocks. In $2019$, Jaberi \cite{Jaberi19} presented an algorithm for computing $2$-twinless-connected components. Jaberi \cite{Jaberi2019} also presented algorithms for computing $2$-twinless blocks.

The paper is organized as follows. In section \ref{def:section1} we present an algorithm for computing $2$-edge-twinless blocks in $O((b_{t}n+m)n)$ time, where $b_{t}$ is the number of twinless bridges.  We then study in section \ref{def:section2} the relation between $2$-edge-twinless blocks and $2$-edge block and we show that $2$-edge-twinless blocks can be computed in $O((b_{t}-b_{s}+n)m)$, where $b_{t}$ and $b_{s}$ are the number of twinless bridges and strong bridges, respectively. Finally, in section \ref{def:openproblems}, we pose some open problems.
\section{Computing $2$-edge-twinless blocks} \label{def:section1}
In this section we present an algorithm for computing the $2$-edge-twinless blocks of directed graphs.
The strongly connected components of a directed graph are disjoint and they can be found in linear time using Tarjan's algorithm \cite{T72}. Moreover, the twinless strongly connected components of a strongly connected graph can be identified in linear time using Raghavan's algorithm \cite{SR06}.
The following lemma shows that we need only consider twinless strongly connected graphs.
\begin{lemma}
Let $G=(V,E)$ be a strongly connected graph. Let $C^{t}$ be a $2$-edge-block of $G$. Then all the vertices of $C^{t}$ lie in the same twinless strongly connected component of $G$.  
\end{lemma}
\begin{proof}
	Assume for the sake of contradiction that $C^{t}$ contains two vertices $v,w$ such that $v,w$ are in distinct twinless strongly connected components $C_{v},C_{w}$, respectively. If we contract every twinless strongly coneccted component of $G$ into a single vertex, we obtain a graph $G^{tss}$. By [\cite{SR06}, Theorem], the underlying graph of $G^{tss}$ is a tree and each edge corresponds to antiparallel edges of $G$. Therefore, there is an edge $e$ such that $G\setminus\left\lbrace e \right\rbrace $ is not strongly connected, a contradiction.
\end{proof}
The $2$-edge-twinless blocks of a strongly connected graph are the union of the $2$-edge-twinless blocks of its twinless strongly connected components. 

The next lemma shows that $2$-edge-twinless blocks have no vertices in common.
\begin{lemma}\label{def:2tbsAredisjoint}
	Let $G=(V,E)$ be a twinless strongly connected graph. The $2$-edge-twinless blocks of $G$ are disjoint.
\end{lemma}
\begin{proof}
	Let $C^{t}_1,C^{t}_2$ be two distinct $2$-edge-twinless blocks of $G$. Suppose for the purpose of contradiction that $C^{t}_1 \cap C^{t}_2 \neq \emptyset$. Let $v\in C^{t}_1 \cap C^{t}_2$ and let $x\in C^{t}_1,y\in C^{t}_2$ with $x,y \notin C^{t}_1 \cap C^{t}_2$.  Since $x,y$ are not in the same $2$-edge-twinless block, there is an edge $e\in E$ such that $x,y$ are not in the same twinless strongly connected component of $G\setminus\left\lbrace e \right\rbrace $. The vertices $x,v$ lie in the same twinless strongly connected component of $G\setminus\left\lbrace e \right\rbrace $ because $C^{t}_1$ is a $2$-edge-twinless block. Furthermore, $v,y$ belong to the same twinless strongly connected component of $G\setminus\left\lbrace e \right\rbrace $.
	 By [\cite{SR06}, Lemma $1$] $x$ and $y$ are in the same twinless strongly connected component of $G\setminus\left\lbrace e\right\rbrace$. 
	
\end{proof}
\begin{lemma}\label{def:Relationbetween2tbAndtbs}
	Let $G=(V,E)$ be a twinless strongly connected graph and let $x,y$ be distinct vertices in $G$. Let $e$ be an edge in $G$ such that $e$ is not a twinless bridge. Then the vertices $x$ and $y$ lie in the same twinless strongly connected component of $G\setminus \lbrace e\rbrace$.
\end{lemma}
\begin{proof} Immediate from the definition.  
\end{proof}
Lemma \ref{def:2tbsAredisjoint} and Lemma \ref{def:Relationbetween2tbAndtbs} lead to an algorithm (Algorithm \ref{def:Relationbetween2tbAndtbs}) which might be helpful when the number of twinless bridges is small. 
\begin{figure}[htbp]
	\begin{myalgorithm}\label{algo:algor2forall2edgetwinlessblocks}\rm\quad\\[-5ex]
		\begin{tabbing}
			\quad\quad\=\quad\=\quad\=\quad\=\quad\=\quad\=\quad\=\quad\=\quad\=\kill
			\textbf{Input:} A twinless strongly connected graph $G=(V,E)$.\\
			\textbf{Output:} The $2$-edge-twinless blocks of $G$.\\
			{\small 1}\> \textbf{If} $G$ is $2$-edge-twinless connected \textbf{then}.\\
			{\small 2}\>\> Output $V$.\\
			{\small 3}\> \textbf{else}\\
			{\small 4}\>\> Let $A$ be an $n\times n$ matrix.\\
			{\small 5}\>\> Initialize $A$ with $1$s.\\
			{\small 6}\>\> compute the twinless bridge of $G$.\\
			{\small 7}\>\> \textbf{for} each twinless bridge $e$ of $G$ \textbf{do} \\
			{\small 8}\>\>\> \textbf{for} each pair $(v,w) \in V\times V$ \textbf{do} \\
			{\small 9}\>\>\>\> \textbf{if} $v,w$ in distinct twinless strongly connected components of $G\setminus \lbrace e\rbrace$ \textbf{then}\\
			{\small 10}\>\>\>\>\> $A[v,w] \leftarrow 0$. \\
			{\small 11}\>\> $E^{t} \leftarrow \emptyset$. \\
			{\small 12}\>\> \textbf{for} each pair $(v,w) \in V\times V $ \textbf{do} \\
			{\small 13}\>\>\> \textbf{if} $A[v,w]=1$ and $A[w,v]=1$ \textbf{then} \\
			{\small 14}\>\>\>\> Add the undirected edge $(v,w)$ to $E^{t}$. \\
			{\small 15}\>\> Calculate the connected components of size $>1$ of $G^{t}$ and output them.
		\end{tabbing}
	\end{myalgorithm}
\end{figure}
\begin{Theorem}\label{def:RunTimeAlgor2ForAll2EdgetwinlessBlocks}
	Algorithm \ref{algo:algor2forall2edgetwinlessblocks} runs in $O((b_{t}n+m)n)$ time, where $b_{t}$ is the number of twinless bridges.
\end{Theorem} 
\begin{proof}
 the twinless strongly connected components of a strongly connected graph can be calculated in linear time using Raghavan's algorithm \cite{SR06}. Jaberi \cite{Jaberi19} shows that twinless bridges can br computed in $O(nm)$ time. Moreover, Lines $7$--$11$ take $O(b_{t} n^{2})$ time.
\end{proof}
\section{The relationship between $2$-edge-twinless blocks and $2$-edge-blocks} 
\label{def:section2}
In this section we discuss the relationship between $2$-edge blocks and $2$-edge-twinless blocks. 
The flowing lemma shows that each $2$-edge-twinless block is a subset of a $2$-edge block.
\begin{lemma}\label{def:relationshipbetween2e2tblocks}
Let $G=(V,E)$ be a twinless strongly connected graph and let $C_{t}$ be a $2$-edge twinless block in $G$. Then all the vertices of $C^{t}$ are in the same $2$-edge block of $G$.
\end{lemma}
\begin{proof}
		Let $v$ and $w $ be distinct vertices in $C_{t}$ and let $e$ be an edge in $G$. By definition, the vertices $v,w$ are in the same twinless strongly connected component of $G\setminus\lbrace e\rbrace$.
		Let $C$ be the twinless strongly connected component of $G\setminus\lbrace e\rbrace$ that contains $v,w$. Clearly, all the vertices of $C$ lie in the same strongly connected component of $G\setminus\lbrace e\rbrace$. Consequently, $v,w$ are in the same $2$-strong block in $G$.
\end{proof}
The next lemma demonstrates which vertices are important when we compute $2$-edge-twinless blocks in $2$-edge-connected graphs.
\begin{lemma}\label{def:2edgeblocks2edgetwinlessblockstwinlessbridges}
	Let $G=(V,E)$ be a $2$-edge connected graph. Let $x,y$ be two distinct vertices in $V$ such that $x,y$ are not in the same $2$-edge-twinless block in $G$. Then $G$ contains a twinless bridge $e$ such that $e$ is not a strong bidge of $G$ and the vertices $x,y$ do not lie in the same twinless strongly connected component of $G\setminus\lbrace e\rbrace$.  
\end{lemma}
\begin{proof}
	Since the vertices $x,y$ are not in the same $2$-edge-twinless block in $G$, there is an edge $(v,w)\in E$ such that $x,y$ are in distinct twinless strongly connected components of $G\setminus\left\lbrace (v,w)\right\rbrace $. Moreover, $(v,w)$ is a twinless bridge because the graph $G\setminus\left\lbrace (v,w)\right\rbrace $ is not twinless strongly connected. But the edge $(v,w)$ is not a strong bridge because $G$ is $2$-edge-connected. 
\end{proof}

Algorithm \ref{algo:analgor2forall2edgetwinlessblock} describes an algorithm for computing $2$-edge-twinless blocks.
\begin{figure}[htbp]
	\begin{myalgorithm}\label{algo:analgor2forall2edgetwinlessblock}\rm\quad\\[-5ex]
		\begin{tabbing}
			\quad\quad\=\quad\=\quad\=\quad\=\quad\=\quad\=\quad\=\quad\=\quad\=\kill
			\textbf{Input:} A twinless strongly connected graph $G=(V,E)$.\\
			\textbf{Output:} The $2$-edge-twinless blocks of $G$.\\
			{\small 1}\> \textbf{If} $G$ is $2$-edge-twinless connected \textbf{then}.\\
			{\small 2}\>\> Output $V$.\\
			{\small 3}\> \textbf{else}\\
			{\small 4}\>\> identify the strong bridges of $G$.\\
			{\small 5}\>\> Compute the twinless bridges of $G$.\\
			{\small 6}\>\> Calculate the $2$-edge blocks of $G$\\
			{\small 7}\>\> \textbf{for} each twinless bridge $e$ of $G$ \textbf{do} \\
			{\small 8}\>\>\> \textbf{if} $e$ is not a strong bridge of $G$ \textbf{then}\\
			{\small 9}\>\>\>\> identify the twinless strongly connected components of $G\setminus\left\lbrace e \right\rbrace $.\\
			{\small 10}\>\>\>\> refine the $2$-edge blocks.
		\end{tabbing}
	\end{myalgorithm}
\end{figure}

\begin{lemma}
	Algorithm \ref{algo:analgor2forall2edgetwinlessblock} calculates $2$-edge-twinless blocks.
\end{lemma}
\begin{proof}
This follows from Lemma \ref{def:relationshipbetween2e2tblocks} and Lemma \ref{def:2edgeblocks2edgetwinlessblockstwinlessbridges}.
\end{proof}

\begin{Theorem}
The running time of Algorithm \ref{algo:analgor2forall2edgetwinlessblock} is $O((b_{t}-b_{s}+n)m)$, where $b_{t}$ and $b_{s}$ are the number of twinless bridges and strong bridges, respectively. 
\end{Theorem}
\begin{proof}
The $2$-edge blocks of a directed graph can be calculated in linear time using the algorithm of Georgiadis et al. \cite{GILP14SODA}. Italiano et al. \cite{ILS12,FILOS12} presented  linear time algorithms for calculating all the strong bridges of a directed graph. Jaberi \cite{Jaberi19} proved that the twinless bridges can be computed in $O(nm)$ time. The twinless strongly connected components of a directed graph can be found in linear time using Raghavan's algorithm \cite{SR06}. Line $9$ can be implemented in linear time by using a similar idea to [\cite{GILP14SODA}, Lemma $3.2$]).
\end{proof}
Italiano et al. \cite{ILS12} proved that the number of strong bridges is at most $2n-2$. Let $G$ be a twinless strongly connected graph. Jaberi \cite{Jaberi19} described how we can obtain a twinless strongly connected subgraph from a strongly connected subgraph in a twinless strongly connected graph without increasing the number of its edges. This means that the number of twinless bridges is at most $2n-2$.
\section{Open Problems} \label{def:openproblems}
We leave as open problem whether the $2$-edge-twinless blocks of a directed graph can be identified in linear time. Another open question is whether the twinless bridges of a directed graph can be calculated in linear time.

Let $G=(V,E)$ be a directed grpah. A $k$-edge-twinless block is a maximal vertex set $U\subseteq V$ such that for any distinct vertices $v,w \in U$ and for each edge subset $L\subseteq E$ with $|L|<K$, the vertices $v,w$ belong to the same twinless strongly connected component of $G\setminus L$.
\begin{lemma}
	The $k$-edge-twinless blocks of a directed graph are disjoint
\end{lemma}
\begin{proof}
Similar to the proof of Lemma \ref{def:2tbsAredisjoint}.
\end{proof}
We leave as open problem whether $k$-edge-twinless blocks can be calculated efficiently.
	
\end{document}